\documentclass{LMCS}

\def\dOi{10(1:15)2014}
\lmcsheading%
{\dOi}
{1--16}
{}
{}
{Jun.~19, 2012}
{Mar.~\phantom03, 2014}
{}


\ACMCCS{[{\bf Theory of computation}]: Computational complexity and cryptography---Complexity theory and logic}
\usepackage{enumitem,hyperref}

 \usepackage{amsmath}
 \usepackage{amssymb}
 \usepackage{amsthm}
 \usepackage{xspace}
 \usepackage{url}
 \usepackage[numbers]{natbib}

\theoremstyle{plain}

\newenvironment{proof*}{\noindent\textit{Proof.}}{}

\newcommand{\Omit}[1]{}
\newcommand{\tup}[1]{\langle #1 \rangle}

\newcommand{\C}{\mathbf{C}}

\newcommand{\set}[1]{\{#1\}}
\newcommand{\setdef}[2]{\left\{#1 : #2 \right\}}

\newcommand{\fop}{\leq_{\text{fop}}}

\renewcommand{\L}{\mathcal{L}}
\renewcommand{\models}{\vDash}
\newcommand{\nmodels}{\nvDash}
\newcommand{\bit}{\textup{BIT}}

\newcommand{\struc}{\textup{Struc}}
\renewcommand{\mod}{\textup{Mod}}

\newcommand{\FO}{\textup{FO}\xspace}
\newcommand{\FOA}{\ensuremath{\textup{FO}\forall}\xspace}
\newcommand{\SO}{\textup{SO}\xspace}
\newcommand{\SOE}{\ensuremath{\textup{SO}\exists}\xspace}

\newcommand{\A}{\mathcal{A}}
\newcommand{\B}{\mathcal{B}}

\newcommand{\cnf}{\textup{CNF}}
\newcommand{\dnf}{\textup{DNF}}

\newcommand{\AC}{\ensuremath{\textup{AC}^0}\xspace}
\newcommand{\LSPACE}{\textup{L}\xspace}
\newcommand{\NLSPACE}{\textup{NL}\xspace}
\newcommand{\PTIME}{\textup{P}\xspace}
\newcommand{\NP}{\textup{NP}\xspace}
\newcommand{\coNP}{\textup{coNP}\xspace}
\newcommand{\NPC}{\textup{NP-complete}\xspace}

\newcommand{\PSPACE}{\textup{PSPACE}\xspace}
\newcommand{\F}{\mathcal{F}}

\newcommand{\HP}{\textsc{HamiltonianPath}\xspace}

\newcommand{\tdm}{\textsc{3dm}\xspace}
\newcommand{\SAT}{\textsc{SAT}\xspace}
\newcommand{\TwoSAT}{\textsc{2SAT}\xspace}

\begin{document}

\title[Universal FO is superfluous for NL, P, NP and co NP]
  {Universal First-Order Logic is Superfluous for NL, P, NP and co NP}

\author[N.~Borges]{Nerio Borges\rsuper a}	
\address{{\lsuper a}Departamento de Matem\'aticas \\
        Universidad Sim\'on Bol\'{\i}var \\ 
        Caracas, Venezuela \\ }	
\email{nborges@usb.ve}  
\thanks{{\lsuper a}This research was partially funded by a S1N grant of the
		Decanato de Investigaci\'on y Desarrollo 
		de la Universidad Sim\'on Bol\'{\i}var.}

\author[B.~Bonet]{Blai Bonet\rsuper b}	
\address{{\lsuper b}Departamento de Computaci\'on \\
        Universidad Sim\'on Bol\'{\i}var \\ 
        Caracas, Venezuela \\ }	
\email{bonet@ldc.usb.ve}  






\begin{abstract}
  \noindent In this work we continue the syntactic study of completeness that began
with the works of Immerman and Medina.
In particular, we take a conjecture raised by Medina in his dissertation
that says if a conjunction of a second-order and a first-order sentences defines an 
NP-complete problems via fops, then it must be the case that
the second-order conjoint alone also defines a NP-complete problem.
Although this 
claim looks very plausible and intuitive, currently we cannot provide a
definite answer for it.
However, we can solve in the affirmative a weaker claim that says that
all ``consistent'' universal first-order sentences 
can be safely eliminated without the fear of losing completeness.
Our methods are quite general and can be applied to complexity classes
other than NP (in this paper: to NLSPACE, PTIME, and coNP), 
provided the class has a complete problem satisfying a certain combinatorial property. 
\end{abstract}

\maketitle

\section{Introduction}\label{sec:intro}

	Descriptive complexity studies the interplay between complexity theory,
finite model theory and mathematical logic. 
Since its inception in 1974 \cite{fagin:spectra}, descriptive complexity
has been able to characterize all the major complexity classes in terms of
logical languages independent of any computational model, thus suggesting
that the computational complexity of languages is a property intrinsic
to them and not an accidental consequence of our choice for the computational
model.

In descriptive complexity, problems are understood as sets of (finite)
models which are described by logical formulas over given vocabularies,
and reductions between problems correspond to first-order definable
functions between the sets of models that characterize the problems.
Like in structural complexity, reductions play a fundamental role in
descriptive complexity, yet unlike the former where the predominant
type of reduction is the many-one reduction, in descriptive complexity
the predominant type of reduction is the first-order projection (fop).
A fop is a weak type of reduction whose study has provided interesting
results such as that common NP-complete problems like \SAT, \HP and
others remain complete via fop reductions, and that such \NP-complete
problems can be described in a \emph{canonical syntactic form}
\cite{borges-bonet12, Dahlhaus,medina:thesis, medina:syntactic}.

The research that led to this paper was motivated by the following conjecture:

\begin{conj}[Medina \cite{medina:thesis}]
\label{conj:medina-full}
Suppose $\Phi$ is a \SOE sentence and $\psi$ is a \FO sentence.
If $\Phi\land\psi$ defines an \NPC problem, then $\Phi$ defines
an \NPC problem as well.
\end{conj}

This is a plausible and intuitive conjecture because it is known that \SOE
captures the class NP while \FO captures (over ordered structures) \AC (i.e., languages recognized
by circuits of polynomial size, constant depth and unbounded fan-in) which
is known to be strictly contained in $\LSPACE \subseteq \PTIME \subseteq \NP$,
and thus the conjunction of $\psi$ with $\Phi$ should not ``add'' hardness
to the problem defined by $\Phi$ because the property defined by $\psi$ is,
by comparison, very easy to check.
In such a case, we say that $\psi$ is superfluous with respect to \NP,
and since $\psi$ is arbitrary, we say that \FO is
superfluous with respect to \NP.

Medina's syntactic study of NP-completeness includes the study of
\emph{syntactic operators} that preserve completeness; i.e., functions
that map sentences into sentences in such a way that a sentence defining
an \NPC problem is mapped into a sentence defining another \NPC problem.
Conjecture~\ref{conj:medina-full} arises in this context since he
observed that superfluous first-order sentences appear in conjunction
with \SOE sentences in the image of \SOE sentences for \NPC problems
for the operator known as \emph{edge creation} \cite{medina:thesis}.

On the other hand, the elimination of such first-order formulas from
conjunctions may prove to be a valuable completeness-preserving operator
in itself, and also provide theoretical justification for the well-known
restriction heuristics that are used to prove the completeness of problems
by enforcing constraints that can be expressed in first-order logic.

In spite of the intuitiveness of this claim, until this date, we cannot
provide a definite answer for it.
However, the main result in this paper implies that the answer
is positive when $\psi$ is a \FOA sentence (i.e., a universal first-order
sentence of the form $\forall\bar x\theta(\bar x)$ where $\theta(\bar x)$
is a quantifier-free formula), and not only for the class \NP but also for
the classes \NLSPACE, \PTIME, \coNP and others.
In general, we show that \FOA is superfluous for all classes that are 
``connected'' in certain way to a combinatorial property that we call
$(n,k)$-uniformity, which up to our knowledge is introduced in this work.

The paper is organized as follows. In the following section, we make the paper
self-contained by revising the necessary definitions from logic, model theory
and descriptive complexity.
Next, we introduce the combinatorial notion of $(n,k)$-uniformity that apply 
to problems, establish some basic properties, and give examples.
Section~\ref{sec:mainresult} proves the main result, which partially solves
the above conjecture in the affirmative by showing that that \FOA
is superfluous for \NP (and other classes).
Section~\ref{sec:applications} shows two applications of our results.
The last section wraps up with a brief summary and
a discussion of future work.

\newpage
\section{Preliminaries}\label{sec:prelim}

\subsection{Logic}

\subsubsection{Syntax}

We consider logical vocabularies without functional symbols  of the form 
$\sigma=\tup{R_1,\ldots,R_r,c_1,\ldots,c_s}$ where each $R_j$ is a relational
symbol of arity $a_j\in\mathbb Z^+$ and each $c_i$ is a constant symbol.
Throughout the paper,
we assume that our logical languages contain the set 
$\set{=,\leq,\bit, suc,0,\max}$
of fixed \emph{numeric} relational and constant symbols, disjoint with
$\sigma$ \cite{immerman:book}.

First-order and second-order formulas over vocabulary $\sigma$ are defined
as usual \cite{Enderton}.  We follow most of the standard notational conventions
found in \cite{Ebbinghaus-Flum, Enderton, immerman:book}. 
The set of all of the first-order (respectively second-order) formulas over
vocabulary $\sigma$ is denoted as $\FO(\sigma)$ (respectively $\SO(\sigma)$).
We write $\FO$ (respectively $\SO$) to denote $\bigcup_{\sigma}\FO(\sigma)$
(respectively $\bigcup_{\sigma}\SO(\sigma)$).
In general if $\L$ is a logic, $\L(\sigma)$ denotes the set of all well formed
formulas in $\L$ over the vocabulary $\sigma$.

An \emph{atomic} formula over vocabulary $\sigma$ has the form $P(t_1,\ldots t_k)$
where $P$ is a $k$-ary relational symbol (numeric or not) and $t_1,\ldots t_k$ are terms
over $\sigma$.
A \emph{literal} is an atomic formula (and then we say it is \emph{positive}) 
or the negation of an atomic formula (and then we say it is \emph{negative}).
A \emph{clause} (respectively an \emph{implicant}) is a disjunction $L_1\lor\cdots\lor L_k$
(respectively conjunction $L_1\land\cdots\land L_k$) of literals.
A {\em numeric formula} in $\L(\sigma)$ is a formula without relational symbols from
$\sigma$, and thus a numeric formula can mention constants from the vocabulary $\sigma$.
Although counterintuitive, this notion guarantees that any non-numeric literal always
refers to a relational symbol in $\sigma$, and also agrees with the notion given in
\cite[Def. 11.7]{immerman:book} where a ``numeric formula'' is one where ``no input
relations occur''.

If $\theta$ is a formula without quantifiers, we say it is in \emph{conjunctive normal form}
(CNF) (respectively \emph{disjunctive normal form} (DNF)) if it is a conjunction (respectively disjunction)
of clauses (respectively implicants).
$k$-CNF (respectively $k$-DNF) is the class of all CNF (respectively DNF) formulas with \emph{at most}
$k$ literals in each clause (implicant). These classes may be ``tagged'' with a vocabulary
$\sigma$ (e.g., $k$-CNF($\sigma$)) when we talk about formulas over the vocabulary $\sigma$.
Likewise, we define $\cnf_k(\sigma)$ as the class of CNF formulas over $\sigma$
whose clauses have at most $k$ non-numeric literals, and similarly for $\dnf_k(\sigma)$.

In general, we use lowercase Greek letters to denote first-order formulas and 
uppercase Greek letter to denote second-order formulas. We write $\psi(x_1,\ldots,x_m)$
to emphasize that the free variables in $\psi$ are among those in
$\tup{x_1,\ldots,x_m}$. A tuple of variables such as $\tup{x_1,\ldots,x_m}$ is
written as $\bar x$ and its length $m$ is denoted by $|\bar x|$. 
A formula with no free variables is referred to as a sentence.
Finally, we use $\Sigma_k^0(\sigma)$ and $\Pi_k^0(\sigma)$ to denote first-order
formulas over vocabulary $\sigma$ with $k$ blocks of alternating quantifiers, beginning
with an existential and universal quantifier respectively. For second-order formulas,
we use the notations $\Sigma_k^1(\sigma)$ and $\Pi_k^1(\sigma)$ respectively.
In all cases, when the vocabulary $\sigma$ is clear from context, we drop it
from the notation.

\subsubsection{Semantics}

Let $\sigma=\tup{R_1,\ldots,R_r,c_1,\ldots,c_s}$ be a vocabulary.
The symbols in $\sigma$ as well as the numeric symbols are interpreted by
$\sigma$-structures.  A (finite) $\sigma$-structure or just a structure,
is a tuple
\[ \A=\tup{|\A|,R_1^\A,\ldots,R_r^\A,c_1^\A,\ldots,c_s^\A} \]
where $|\A|$ is a symbol that denotes the universe (or domain) of $\A$,
each  $R_j^\A\subseteq|\A|^{a_j}$ is a $a_j$-ary relation over $|\A|$ and
each $c_j\in|\A|$ is an element of $|\A|$. 
The number of elements in the universe, size or cardinality of $\A$ is
denoted by $\|\A\|$.
Following \cite{immerman:book}, we assume $|\A|$ to be an initial segment
of size greater than 1 of the set of the natural numbers; i.e., 
$|\A|=[n]=\set{0,1,\ldots, n-1}$ with $n>1$.
We say that $R_j^\A$ and $c_j^\A$ are the \emph{interpretations} of the
relational and constant symbols $R_j$ and $c_j$ in the structure $\A$.
On the other hand, the numeric symbols obtain the \emph{standard interpretations}
in the fragment $[n]$ of the natural numbers \cite{immerman:book}; e.g., the
symbols `$=$' and `$\leq$' are interpreted by the equality and non-strict
natural order in $\mathbb N$ respectively, while the relational symbol
$\bit$ is interpreted by the binary relation $\bit^\A$ given by 
\[ (i,j)\in \bit^\A\iff\text{ the $j$-th bit in the binary expansion of $i$ is 1}\,, \]
where the zeroth position is the least significant bit of $i$.
The set of all the finite $\sigma$-structures is denoted by $\struc(\sigma)$.

We follow the standard definition for the relation $\models$ between structures
and first-order formulas, and its extension to second-order formulas (see, e.g.,
\cite{Ebbinghaus-Flum,Enderton,immerman:book}). These relations are defined in
terms of the relation $\models$ for pairs $\tup{\A,i}$ and formulas $\psi$,
where $\A$ is a structure and $i$ is an interpretation of variables into the
elements of $\A$.
If $\varphi(x_1,\ldots,x_m)$ is a formula over the vocabulary $\sigma$,
$\A\in\struc(\sigma)$, and $\tup{a_1,\ldots,a_m}$ is a tuple over $|\A|$,
then $\A\models\varphi(a_1,\ldots,a_m)$ means that $\tup{\A,i}\models\varphi$
for every interpretation $i$ that maps $x_j$ into $a_j$ for $1\leq j\leq m$.
In particular, when $\varphi$ is a sentence, $\A\models\varphi$
iff $\tup{\A,i}\models\varphi$ for every interpretation $i$.

Given a vocabulary $\sigma$, we say that a $\sigma$-structure $\A$ is a
\emph{model} of sentence $\phi$ if $\A\models\phi$.
The set of all the finite models of $\phi$ is denoted by $\mod(\phi)$;
notice that $\mod(\phi)\subseteq\struc(\sigma)$.
Similarly for second-order formulas.

\subsection{Decision Problems and Complexity Classes}

A \emph{decision problem} (or just \emph{problem}) $S$ is a subset of
$\struc(\sigma)$ for some fixed $\sigma$, which is closed under isomorphisms.
For example, the problem {\sc ThreeDimensionalMatching (3dm)} can be
thought as the set of all structures $\A=\tup{|\A|,M^\A}$ over the
vocabulary $\sigma=\tup{M^3}$, where $M$ is a ternary relational symbol,
such that $M^\A$ contains a 3-dimensional matching; i.e., $M^\A$ contains
a set of triplets
\[ M' = \set{(a_0,b_0,c_0),\ldots, (a_{\|\A\|},b_{\|\A\|},c_{\|\A\|})} \]
such that $a_i\neq a_j$, $b_i\neq b_j$ and $c_i\neq c_j$ for every $i\neq j$.

If $S\subseteq\struc(\sigma)$ is a decision problem, then every finite
$\sigma$-structure $\A$ is an \emph{instance} and every element of $S$
is a \emph{positive instance} of $S$.
We say that a Turing Machine $M$ \cite{Garey&Johnson, Papadimitriou}
\emph{decides} a problem $S$ if, given a suitable encoding of an instance
$\A$ of $S$ as input to $M$, $M$ accepts the input iff $\A$ is a positive
instance of $S$.

Decision problems are classified into \emph{complexity classes}
accordingly to their difficulty.
We assume the standard computational resources (time and space) and
computational modes (deterministic and non-deterministic) found in
the literature \cite{Garey&Johnson,immerman:book,Papadimitriou}.
The most important complexity classes \cite{immerman:book, Papadimitriou}
are \LSPACE, \NLSPACE, \PTIME, \NP and \PSPACE. The following
chain of inclusions is a well-known fact:
\[ \LSPACE \subseteq \NLSPACE \subseteq \PTIME \subseteq \NP \subseteq \PSPACE \,. \]
\LSPACE and \NLSPACE are respectively the classes of all problems solvable by deterministic
and non-deterministic Turing machines that use logarithmic space, \PTIME and \NP
are respectively the classes of problems solvable by deterministic and non-deterministic 
Turing machines in polynomial time, and \PSPACE is the class of problems solvable
by deterministic Turing machines that use polynomial space
\cite{Garey&Johnson,immerman:book,Papadimitriou}.
Given a complexity class $\C$, the class $\text{co}\C$ is the class of all
problems whose complement is in $\C$.
For example, for a problem $S\subseteq\struc(\sigma)$ in the class $\C$,
the problem $\overline{S}=\struc(\sigma)\setminus S$ belongs to $\text{co}\C$.
Turing machines can also compute functions \cite{Papadimitriou}, and similar
complexity measures apply to such functions.

A decision problem is typically characterized by a sentence over some logic $\L$.
\tdm, for example, is characterized by the sentence:
\begin{alignat*}{1}
\Phi_\tdm :=
    \exists R^3 \bigl[\,
        &\forall \bar x (R(\bar x)\longrightarrow M(\bar x)) \,\land \\
        &\forall x (\exists x_2 x_3  R(x,x_2,x_3) \land \exists x_1 x_3 R(x_1,x,x_3) \land \exists x_1 x_2 R(x_1,x_2,x)) \,\land\\
        &\forall \bar x\bar y (R(\bar x)\land R(\bar y)\land \bar x\not=\bar y
		\longrightarrow x_1\neq y_1\land x_2\neq y_2\land x_3\neq y_3)
    \, \bigr] 
\end{alignat*}
in which the quantified ternary relation $R$ denotes the 3-dimensional matching
contained in the input instance whose triplets are given by $M$. Hence, if $\Phi_\tdm$
is satisfied by a structure $\A$, then the set of triplets $M^\A$ that interpret
the relational symbol $M$ contains a 3-dimensional matching.
Thus, the problem \tdm corresponds to the class $\mod(\Phi_\tdm)$ of finite
$\sigma$-structures $\A$ such that $\A\models\Phi_\tdm$.
We also say that $\Phi_\tdm$ \emph{defines} \tdm in \SOE and that \tdm is
\emph{definable} in \SOE.

In general, for a logical language $\L$ and complexity class $\C$, we write
$\C\leq \L$ if every problem in $\C$ is definable in the logic $\L$.
On the contrary, if $\mod(\phi)$ is a problem in $\C$ for every sentence $\phi$
in $\L$, we write $\L\leq\C$.  If both $\L\leq\C$ and $\C\leq\L$ hold, we say
that the logic $\L$ \emph{captures} $\C$ and write $\L=\C$.
The logic $\L$ {\em captures the complexity class $\C$ over 
the class of structures $K$} if for every problem $A$ in $\C$ there is
a sentence $\phi\in \L$ such that $\mod(\phi)= A\cap K$ and
for every sentence $\phi\in L$: $(\mod(\phi)\cap K)\in \C$.

It is known that $\SOE=\NP$ \cite{fagin:spectra}, and $\SOE\text{-Horn}=\PTIME$
and $\SOE\text{-Krom}=\NLSPACE$ over ordered structures \cite{Gradell}.
See the textbook of Immerman \cite{immerman:book} for other characterizations
and results.

\subsection{First-Order Queries and Projections}

\subsubsection{Reductions and Completeness}

The idea of reduction is fundamental in complexity theory.
Roughly speaking, we say that a problem $S\subseteq\struc(\sigma)$ 
\emph{reduces} to a problem $T\subseteq\struc(\tau)$ if there is a
function $f:\struc(\sigma)\rightarrow\struc(\tau)$ such that, for
every $\sigma$-structure $\A$, $f(\A)$ is a positive instance of $T$
if and only if $\A$ is a positive instance of $S$; in such a case,
we say that $f$ is a \emph{reduction} from $S$ to $T$.
Informally, this means that $S$ is as hard to solve as $T$ provided
that $f$ is relatively ``easy'' to compute.
 
Reductions are classified according to their computation complexity.
If there is a function $f$ of type $r$ reducing $S$ to $T$, the we
write $S\leq_r T$ and say that $S$ \emph{reduces to $T$ via $r$-reductions}.
Given a complexity class $\C$ we say that the problem $T$ is \emph{$\C$-complete via $r$-reductions}
if $T$ belongs to $\C$ and $S\leq_r T$ for every problem $S$ in $\C$. 
If the type of the reductions is clear from the context, we just say
that $T$ is $\C$-complete.

\subsubsection{First-Order Projections}

In this paper we are concerned with a very simple type of reduction
that is definable in first-order logic and called first-order projections.
We describe it in the following, but first need to define first-order queries.

Let $\sigma$ and $\tau=\tup{R_1^{a_1},\ldots,R_r^{a_r},c_1,\ldots,c_s}$
be two vocabularies, $k\geq1$ be an integer, and consider the tuple
$I=\tup{\varphi_1,\ldots,\varphi_r,\psi_1,\ldots,\psi_s}$ of $r+s$
first-order formulas in $\FO(\sigma)$ of the form $\varphi_i(x_1,\ldots,x_{ka_i})$
for $1\leq i\leq r$, and $\psi_j(x_1,\ldots,x_k)$ for $1\leq j\leq s$.
That is, $\varphi_i$ has at most $ka_i$ free variables among those in $\set{x_1,\ldots,x_{ka_i}}$,
and $\psi_j$ has at most $k$ free variables among those in $\set{x_1,\ldots,x_k}$.

The tuple $I$ defines a mapping $\A\mapsto I(\A)$, called a \emph{first-order query}
of arity $k$, from $\sigma$-structures into $\tau$-structures.
For given $\A\in\struc(\sigma)$, the map $\A\mapsto I(\A)$ is given by:
\begin{enumerate}
\item Universe $|I(\A)|$ defined as $|\A|^k$ (the $k$-tuples over $|\A|$),
\item Relations $R_i^{I(\A)}\doteq\{(\bar{u}_1,\ldots,\bar{u}_{a_i})\in|\A|^{ka_i}:
                                    \A\models\varphi_i(\bar{u}_1,\ldots,\bar{u}_{a_i})\}$, and
\item Constants $c_j^{I(\A)}\doteq\bar{u}$ for the \emph{unique} $\bar{u}$
  with $\A\models\psi_j(\bar{u})$. If there is no such $\bar u$ or there is
  more than one, the query is not well defined.
\end{enumerate}
Notice that we defined $|I(\A)|$ as $|\A|^k$. In order to keep our convention
of $|I(\A)|$ being an initial segment of $\mathbb N$, we identify it with
the set $\set{0,\ldots,\|\A\|^k-1}$ by lexicographically ordering $|\A|^k$.
The numeric relations and constants in $I(\A)$ are defined in the standard
way such that the numeric symbols obtain the intended interpretations.
It is not difficult to show that the formulas defining the numeric predicates
are all first-order formulas \cite{immerman:book}. 
Some authors consider mappings $I$ extended with a formula $\varphi_0$ used
to define the universe as $|I(\A)|=\{\bar u\in|A|^k:\varphi_0(\bar u)\}$.
This however causes difficulties when defining the interpretation
of the numeric predicates as, in some cases, the formulas defining them
cease to be first-order \cite{immerman:book}. For this reason, we do not
consider such formulas $\varphi_0$.

If $S\subseteq\struc(\sigma)$ and $T\subseteq\struc(\tau)$ are two problems,
and the query $I$ is such that $\A\in S$ iff $I(\A)\in T$, then $I$ is called
a \emph{first-order reduction} from $S$ to $T$.
A first-order query is called a \emph{first-order projection} (fop) if each
$\varphi_i$, and each $\psi_j$, has the form
\[ \alpha_0(\bar x) \lor (\alpha_1(\bar x) \land \lambda_1(\bar x)) \lor
   \cdots \lor (\alpha_e(\bar x) \land \lambda_e(\bar x)) \]
where the $\alpha_k$'s are numeric and pairwise \emph{mutually exclusive}, 
and each $\lambda_k$ is a $\sigma$-literal.
Two formulas $\alpha(\bar x)$ and $\beta(\bar x)$ over vocabulary $\sigma$ are
mutually exclusive if, given any finite $\sigma$ structure $\A$ and tuple
$\bar a\in|\A|^{|\bar x|}$, it holds $\A\models \neg\alpha(\bar a)\lor\neg\beta(\bar a)$.
Projections are typically denoted by the letter $\rho$.
If $S$ is complete for the class $\C$ via $\fop$ reductions, then we say
that $S$ is \emph{$\C$-complete via fops} or simply \emph{$\C$-complete}.

\section{Definitions and Basic Facts}
\label{sec:basics}

This section contains most of the new definitions that we will need
throughout the rest of the paper and some basic properties.
We begin by extending Medina's notion of superfluity \cite{medina:thesis}:

\begin{defi}[Superfluity]
Let $\sigma$ and $\tau$ be two vocabularies, $\L$ be a logic and $\C$ be a complexity
class captured by $\L$.  Then,
\begin{enumerate}
\item A sentence $\psi$ in $\L$ is \emph{superfluous} with respect to fop
  $\rho:\struc(\sigma)\rightarrow\struc(\tau)$ if $\rho(\A)$ satisfies $\psi$
  for every finite $\sigma$-structure $\A$.
\item A sentence $\psi$ in $\L$ is \emph{superfluous} with respect to $\L$ if
  for every sentence $\Phi$ in $\L$:
  \[ \text{$\mod[\Phi\land\psi]$ is $\C$-complete} \implies \text{$\mod[\Phi]$ is $\C$-complete} \,. \]
\item A fragment $\L'\subseteq \L$ is \emph{superfluous} with respect to $\L$ 
  if every sentence $\psi$ in $\L'$ is superfluous with respect to $\L$.		
\item A fragment $\L'\subseteq \L$ is \emph{superfluous} with respect to
  a complexity class $\C$ if $\C$ is captured by $\L$ and $\L'$ is superfluous 
  with respect to $\L$.
\end{enumerate}
\end{defi}
\noindent We will use these definitions to successively break down the problem
of showing that \FOA is superfluous with respect to a complexity class
like \NP into the simpler problem of showing that an implicant $\psi$
is superfluous with respect to a fop-reduction $\rho$.
Indeed, the problem of showing that \FOA is superfluous for \NP is
reduced to showing that \FOA is superfluous for \SOE, that is again
reduced to showing that every formula $\forall\bar x\psi(\bar x)$,
where $\psi(\bar x)$ is in $\cnf_r$ for some $r$, is superfluous
for \SOE, which is ultimately reduced to showing that $\psi(\bar x)$
is superfluous with respect to a fop $\rho$.
With the terminology offered in this definition, the conjecture at
the beginning of this paper can be rephrased as saying that \FO
is superfluous with respect to \NP.

In the following, we develop the notions of $n$-consistency of a formula
and $(n,k)$-uniformity of a subset of structures. These notions will play
a fundamental role in our methods.

\begin{defi}
Let $\sigma$ be a vocabulary, $\varphi(\bar x)$ be a formula in $\FO(\sigma)$,
$n$ be a natural number, and $\bar u\in [n]^{|\bar x|}$ be a tuple of natural
numbers.
We say that $\tup{\varphi(\bar x),\bar u}$ is \emph{$n$-consistent} if
there is a $\sigma$-structure $\A$ with $\|\A\|=n$ such that
$\A\models\varphi(\bar u)$.
If $S$ is a subset of finite $\sigma$-structures, we sat that
$\tup{\varphi(\bar x),\bar u}$ is \emph{$n$-consistent in $S$} if
there is a $\sigma$-structure $\A$ in $S$ with $\|\A\|=n$ such that
$\A\models\varphi(\bar u)$.
When there is no risk of confusion, we abbreviate by
just saying that $\varphi(\bar u)$ is $n$-consistent (in $S$).
\end{defi}

\begin{defi}[Uniformity]
\label{def:uniform}
Let $\sigma=\tup{R_1^{a_1},\ldots,R_s^{a_s},c_1,\ldots,c_t}$ be a vocabulary,
and $S$ be a subset of finite $\sigma$-structures.
Let $n$ and $k$ be two natural numbers.
We say that $S$ is \emph{$(n,k)$-uniform} iff for
\begin{itemize}
\item every integer $m\geq n$ and non-negative integers $p$ and $q$ such that $p+q\leq k$,
\item every sequence $L_1(\bar t_1),\ldots,L_p(\bar t_p)$ of $\sigma$-literals,
\item every sequence $\bar u_1,\ldots,\bar u_p$ of tuples with $\bar u_j\in [m]^{|{\bar t}_j|}$ for $1\leq j\leq p$,
\item every sequence $c_{t_1},\ldots,c_{t_q}$ of constant symbols in $\sigma$, and
\item every sequence $b_1,\ldots, b_q$ of integers with $b_j \in [m]$ for $1\leq j\leq q$, the following holds:
\end{itemize}
If $\varphi(\bar u,\bar b):=\bigwedge_{j=1}^p L_j(\bar u_j) \land \bigwedge_{j=1}^q c_{t_j}=b_j$
is $m$-consistent, then it is also $m$-consistent in $S$.
\end{defi} 

%
%
%
%
%
%

%
%

The following properties follow directly from the definition.

\begin{lem}
\label{le:prop1-uniformity}
\label{le:monotonicity}
If $S\subseteq\struc(\sigma)$ is $(n,k)$-uniform, then it is also $(n,k-1)$-uniform
and $(n+1,k)$-uniform. If $S$ is $(n,k)$-uniform and $S\subseteq T$, then
$T$ is $(n,k)$-uniform.
\end{lem}

Finally, we provide some examples of $(n(k),k)$-uniform problems, where $k$
ranges over positive integers and $n(k)$ is an increasing function from
$\mathbb Z^+$ to itself:

\begin{lem}\label{le:uniformity-examples}
For every positive integer $k$:
\begin{enumerate}
\item {\sc Reach} is $(2k+1,k)$-uniform,
\item {\sc AltReach} is $(2k+1,k)$-uniform,
\item {\sc HamiltonianPathBetweenZeroAndMax (0m-HP)} is $(4k,k)$-uniform, and
\item {\sc coMonoTriangle} is $(2k+6,k)$-uniform.
\end{enumerate}
\end{lem}

Here, {\sc Reach} and {\sc AltReach} refer to graphs $G$ in which there is a path
from a designated vertex $s$ to a designated vertex $t$, the difference among the
two problems being that the first refers to regular graphs while the second to
alternating graphs. 
An \emph{alternating graph} is a directed graph whose vertices are
partitioned into \emph{universal} and \emph{existential}.
The accessibility relation or alternating paths are defined as follows.
Any vertex $u$ is accessible from itself.
If $u$ is an existential vertex, there is an edge $(u,v)$ and $w$ is a vertex 
accessible from $v$, then $w$ is accessible from $u$. Finally, if $u$ is an
universal vertex, there is at least one edge leaving $u$ and $w$ is accessible
from every vertex $v$ such that $(u,v)$ is an edge, then $w$ is accessible
from $u$. 

{\sc 0m-HP} refers to graphs in which there is a Hamiltonian
path between the vertex denoted by the constant 0 and the vertex denoted by 
the constant $\max$. {\sc coMonoTriangle} is the complement of the problem
{\sc MonochromaticTriangle}; i.e., {\sc coMonoTriangle} consists of all graphs
$G$ for which every 2-coloring of the edges of $G$ contains a monochromatic triangle.
Formal definitions for these problems are provided in the Appendix.

%
%
%
%
%
%

\begin{proof}
\textbf{Case 1.}
{\sc Reach} is defined with the vocabulary $\sigma=\tup{E,s,t}$ where $E$ a
binary relation and $s,t$ are two constant symbols.
Let $\varphi$ be an implicant formula like the one in Definition~\ref{def:uniform}
that has at most $k$ literals.
Such a formula postulates the presence/absence of at most $k$ edges 
in the graph encoded by a given $\sigma$-structure. The idea of the 
proof is to construct a path from $s$ to $t$ by adding two new
edges $(s,a')$ and $(a',t)$ for a vertex $a'$ that is not referred to
by the formula $\varphi$, which exists since the graph is assumed
to have $2k+1$ vertices and $\varphi$ can refer to at most $2k$
different vertices.
Formally, notice that that $\varphi$ has at most $2k$ free-variables
that we denote by $\bar x$ and write $\varphi(\bar x)$.
Let $k\geq 0$ and $m\geq 2k+1$ be two integers and consider a
sequence $\bar a$ of $2k$ integers in $[m]$ such that
$\varphi(\bar a)$ is $m$-consistent. 
Then, there is a model $\A=\tup{[m],E^\A,s^\A,t^\A}$ of size $m$ that
satisfies $\varphi(\bar a)$.
Let $a'\in[m]$ be an integer not contained in the sequence $\bar a$,
which exists since $m\geq 2k+1$. Then, the structure $\A'=\tup{[m],E^{\A'},s^\A,t^\A}$,
where $E^{\A'}=E^\A \cup \{(s^\A,a'),(a',t^\A)\}$, satisfies $\varphi(\bar a)$ and
belongs to {\sc Reach}.
Therefore, {\sc Reach} is $(2k+1,k)$-uniform for every $k\geq 0$.
                                  
\medskip\noindent
\textbf{Case 2.}
{\sc AltReach} is defined with the vocabulary $\sigma=\tup{E,U,s,t}$
where $E$ denotes the edges of the graph and $U$ is a unary relation
denoting the universal vertices.
As before, let $\varphi$ be an implicant with at most $k$ literals.
Such a formula refers to at most $2k$ vertices and thus can be
written as $\varphi(\bar x)$ where $|\bar x|=2k$.
Let $k\geq 0$ and $m\geq 2k+1$ be two integers and consider a
sequence $\bar a$ of $2k$ integers in $[m]$ such that
$\varphi(\bar a)$ is $m$-consistent. 
Let $\A=\tup{[m],E^\A,U^\A,s^\A,t^\A}$ be a \emph{minimum}
$\sigma$-structure of size $[m]$ that satisfies $\varphi(\bar a)$,
where the minimum is with respect to the sizes of $E^\A$ and $U^\A$
(i.e., $\A$ satisfies $\varphi(\bar a)$ but if some edge from
$E^\A$ or some vertex from $U^\A$ is removed, then $\A$
ceases to satisfy $\varphi(\bar a)$).
It is not hard to see that in the graph encoded by $\A$, there
are at most $2k$ vertices connected to $s$ through simple (non-alternating)
paths.
Let $a'$ be an integer in $[m]$ not contained in $\bar a$.
We form a new structure $\A'$ that is like $\A$ but with the
additional edges
\[ \{ (s^\A,a'), (a',t^\A) \} \cup \{ (a_i,a') : 1\leq i\leq 2k \} \,. \]
It is not hard to see that $\A'$ satisfies $\varphi(\bar a)$ and has an
alternating path connecting $s$ to $t$. The latter because every vertex
$a_i$ connected to $s$ through a simple path in $\A$ is connected to
$t$ through an alternating path, which then implies that every such
vertex $a_i$ and $s$ itself is connected to $t$ through an alternating
path.
Therefore, {\sc AltReach} is $(2k+1,k)$-uniform.

\medskip\noindent
\textbf{Case 3.}
{\sc 0m-HP} is defined with the vocabulary $\sigma=\tup{E}$ consisting of
a single binary relation symbol $E$.
As before, the formula $\varphi$ refers to at most $2k$ vertices
and can be written as $\varphi(\bar x)$ with $|\bar x|=2k$.
Let $k\geq 0$ and $m\geq 4k$ be two integers and consider a 
sequence $\bar a$ of $2k$ integers in $[m]$ such that $\varphi(\bar a)$
is $m$-consistent.
Let $\A=\tup{[m],E^\A}$ be a $\sigma$-structure satisfying $\varphi(\bar a)$.
We need to construct a structure $\A'$ of the same size that 
also satisfies $\varphi(\bar a)$ and has a Hamiltonian path
between $0$ and $\max=m-1$.

Let $V=\{a_1,a_2,\ldots,a_{2k}\}\setminus\{0,m-1\}$ be the set of 
\emph{unique} elements in the sequence $\bar a$ except $0$ and $m-1$.
The set $V$ has $\ell\leq 2k$ elements and thus we can pick another
$\ell$ elements $V'=\{c_1,\ldots,c_{\ell}\}$ such that $V\cap V'=\emptyset$.
We construct $\A'$ like $\A$ but add edges so to construct the path
\[ 0c_1a'_1c_2a'_2 \ldots c_\ell a'_\ell \]
that begins at 0 and does not repeat any vertex. This path contains
$2\ell+1\leq m$ vertices. If there are vertices $w_1,\ldots,w_q$ in $[m]$
that do not lie on the path, the path can be ``completed'' into
a Hamiltonian path joining $0$ and $m-1$ of the form
\[ 0c_1a'_1c_2a'_2 \ldots c_\ell a'_\ell w_1 w_2 \ldots w_q (m-1) \,. \]
This path contains all the elements in $[m]$ and does not violate
any of the conditions imposed by $\varphi(\bar a)$.
Therefore, there is a positive instance of {\sc 0m-HP} that
satisfies $\varphi(\bar a)$, and {\sc 0m-HP} is thus $(4k,k)$-uniform
for every $k\geq 0$.

\medskip\noindent
\textbf{Case 4.}
{\sc coMonoTriangle} is also defined with the vocabulary $\sigma=\tup{E}$.
Consider a formula $\varphi(\bar x)$ as the above for {\sc 0m-HP}.
Let $k\geq 0$ and $m\geq2k+6$ be two integers and consider a sequence
$a_1,\ldots,a_{2k}$ of integers in $[m]$ such that $\varphi(\bar a)$
is $m$-consistent. Let $\A=\tup{[m],E^\A}$ be a $\sigma$-structure
satisfying $\varphi(\bar a)$.
Since $m\geq 2k+6$, there are 6 vertices $\{u_0,\ldots,u_5\}$ that
are not mentioned in $\bar a$. Then, the structure $\A'=\tup{[m],E^{\A'}}$
where $E^{\A'}=E^\A\cup\{(u_i,u_j):0\leq i,j\leq 5\}$ is a structure
that satisfies $\varphi(\bar a)$.
However, $\A'$ encodes a graph $G$ that has a complete subgraph of order
6 and therefore every 2-coloring of the edges of $G$ contains a monochromatic
triangle\footnote{This is a well-known result in Ramsey theory, but follows
  easily using the pigeonhole principle \cite{ramsey}. Consider a red/blue
  coloring of the edges of a $K_6$ and pick any vertex $v$. There are 5
  edges incident at $v$, so by the pigeonhole principle there are at least
  three edges of the same color: say blue and incident at the vertices
  $\{x,y,z\}$.  If any edge of the edges connecting $\{x,y,z\}$ is blue,
  then there is a blue triangle. Otherwise, all such edges are red and
  there is a red triangle.
}
and $\A'$ belongs to {\sc coMonoTriangle}.
Hence {\sc coMonoTriangle} is $(2k+6,k)$-uniform for every $k\geq0$.
\end{proof}

Notice that $(n,k)$-uniformity is not preserved under complementation.
For example, it is easy to see that {\sc MonoTriangle} is not $(n,k)$-uniform
for any pair of natural numbers $(n,k)$ with $k\geq 15$ since, by a
result of Ramsey \cite{ramsey}, no graph in {\sc MonoTriangle} satisfies
$\bigwedge\setdef{E(i,j)}{0\leq i<j\leq 5}$.

\section{Main Result}
\label{sec:mainresult}

In this section we state and prove the main result of this paper
which, among other things, implies the superfluity of \FOA with 
respect to \NP. Basically, the result says that if a complexity
class $\C$ contains a family of complete problems that are
$(n,k)$-uniform for increasing values of $k$, plus other 
conditions, then \FOA is superfluous with respect to $\C$.
The proof is a direct consequence of a series of results
that are presented in the form of one proposition and two
lemmas.

The proposition is a standard result in logic whose proof is left
to the reader. It is used in the proof of Lemma~\ref{le:corelemma1}
to show that a numeric formula $\gamma(\bar x)$ holds in a
structure $\A'$ when it holds in a structure $\A$ of the same size
and with the same interpretation for the constant symbols.
Remember that a numeric formula is one with no relation symbols
from the vocabulary.

\begin{prop}
\label{prop:constantslemma}
Let $\sigma$ be a vocabulary with constants $c_1,\ldots, c_t$ and
$\varphi(\bar x)\in\FO(\sigma)$ be a numeric formula.
If $\A$ and $\A'$ are two finite $\sigma$-structures of the same
size and $c_j^\A=c_j^{\A'}$ for every $1\leq j\leq t$, then for every
interpretation of variables $i$:
$\tup{\A,i}\models\varphi(\bar x)$ iff $\tup{\A',i}\models\varphi(\bar x)$.
\end{prop}

The following two lemmas are novel and related to the concepts of 
consistency and uniformity presented above.
The first establishes the consistency of the image (of a fop $\rho$)
for a uniform subset $S$ of structures with respect to formulas
$\psi(\bar x)$ that are \emph{implicants}.
This result is then used in the second lemma to establish that \FOA
formulas $\psi=\forall\bar x\theta(\bar x)$ are superfluous with
respect to reductions $\rho$ that map a uniform subset $S$ into
$\mod(\Phi\land\psi)$ for any sentence $\Phi$.

\begin{lem}\label{le:corelemma1}
Let $\sigma$ and $\tau=\tup{R_1,\ldots,R_s,c_1,\ldots,c_t}$ be two
vocabularies where each $R_j$ is an $a_j$-ary relation symbol and
$\sigma$ has constants $\{c'_1,\ldots,c'_{t'}\}$, and let $n\geq 0$
and $k\geq t+t'$ be two non-negative integers.
Further, let 
\begin{description}
  \item[H1] $\rho:\struc(\sigma)\rightarrow\struc(\tau)$ be a $d$-ary fop given by the
	tuple $\tup{\varphi_1,\ldots,\varphi_s,\psi_1,\ldots,\psi_t}$ consisting of
	projective formulas,
  \item[H2] $S\subseteq\struc(\sigma)$ be an $(n,k)$-uniform subset of $\sigma$-structures,
  \item[H3] $\psi(\bar x)$ be a $\dnf_{k-t-t'}(\tau)$ implicant,
  \item[H4] $\A$ be a finite $\sigma$-structure with $\|\A\|\geq n$, and
  \item[H5] $\bar u$ be a tuple of elements from $|\rho(\A)|$ such that $\rho(\A)\models\psi(\bar u)$.
\end{description}
Then, there is a structure $\A'$ in $S$ with $\|\A\|=\|\A'\|$ such that $\rho(\A')\models\psi(\bar u)$.
\end{lem}

\begin{proof*}
The strategy for the proof is as follows. First, we give a claim that 
relates statements about $\rho(\A)$ in terms of statements about $\A$,
and vice versa.
Second, we use the claim to construct a formula $\delta(\bar x)$ and
tuple $\bar a$ over $|\A|$ such that $\A\models\delta(\bar a)$ iff
$\rho(\A)\models\psi(\bar u)$ plus conditions that guarantee that the 
constants are interpreted in the same way in the structures $\A$ and $\A'$,
and in $\rho(\A)$ and $\rho(\A')$ respectively.
Third, we use the uniformity of $S$ to obtain a structure $\A'\in S$
that satisfies $\delta(\bar a)$.  Finally, we use the claim again, but
applied to the structure $\A'$, to show $\rho(\A')\models\psi(\bar u)$.

For the claim, recall that the mapping $\B\mapsto\rho(\B)$ is defined by
the projective formulas that define the interpretations in $\rho(\B)$ of the
relation and constant symbols in $\tau$; i.e., $\rho(\B)\models R_j(\bar a)$
iff $\B\models\varphi_j(\bar a)$ where $\bar a$ is a tuple in $|\B|^{da_j}$
(because $\rho$ has arity $d$ and $R_j$ has arity $a_j$), and
$\rho(\B)\models c_j=\bar w$ iff $\B\models\psi_j(\bar w)$ where $\bar w$ is
a tuple in $|\B|^d$. So, whether we are dealing with formulas interpreted
in $\B$ or $\rho(\B)$, the tuples $\bar a$ are always over elements in $|\B|$.

\medskip\noindent
\begin{clm}\footnote{This claim is a special case of a more general
  result expressed in terms of the dual operator associated with a
  first-order query \cite{immerman:book}.}

Let $\B$ be a $\sigma$-structure, $\eta(\bar y)$ be either a non-numeric
$\tau$-literal or an atomic formula of the form $c=\bar y$, and $\bar a$
be a tuple of $|\bar y|$ elements in the universe $|\B|$.
Then, there is a formula $\mu(\bar y)$ over $\sigma$ that is either numeric
or a conjunction $\alpha(\bar y)\land\lambda(\bar y)$, with $\alpha$ being
numeric and $\lambda$ being a non-numeric $\sigma$-literal, such that
$\rho(\B)\models\eta(\bar a)$ iff $\B\models\mu(\bar a)$.
\end{clm}

\bigskip\noindent
We now use the claim to prove the lemma. The proof of the claim appears at
the end of this proof.

By H3, the implicant has form $\psi(\bar x)=\theta(\bar x) \land L_1(\bar x)\land \cdots \land L_m(\bar x)$
where $\theta(\bar x)$ is a numeric implicant and $L_1,\ldots,L_m$ are non-numeric
$\tau$-literals with $m\leq k-t-t'$.
By H5, we know that
\begin{displaymath}
\rho(\A)\models\theta(\bar u),\quad \rho(\A)\models L_1(\bar u),\quad \rho(\A)\models L_2(\bar u),\quad \cdots\quad \rho(\A)\models L_m(\bar u)  \,.
\end{displaymath}

Let us apply the Claim to these entailments and the tuple $\bar u$ to obtain
the formulas
\begin{displaymath}
\mu_1(\bar y),\,\mu_2(\bar y),\,\ldots,\, \mu_m(\bar y)
\end{displaymath}
such that $\A\models\mu_j(\bar u)$ for $1\leq j\leq m$.
Thus, $\A\models \mu_1(\bar u)\land \cdots \land \mu_{m}(\bar u)$.
By collecting the numeric subformulas in the $\mu_j$'s into a single numeric formula
$\beta(\bar y)$, we obtain
\begin{displaymath}
\A\models \beta(\bar u)\land \lambda_1(\bar u)\land\cdots\land\lambda_\ell(\bar u)
\end{displaymath}
where each $\lambda_j$ is a non-numeric $\sigma$-literal for $1\leq j\leq \ell$
with $\ell\leq m$. At this stage, we can apply the uniformity of $S$ to obtain
a structure $\A'\in S$ satisfying $\lambda_j(\bar u)$ for $1\leq j\leq \ell$; yet
this is not enough as we also require $\A'\models\beta(\bar u)$ and
$\rho(\A')\models\theta(\bar u)$.
However, Proposition~\ref{prop:constantslemma} can help us provided that $\A$ and $\A'$
interpret the constants in the same way, and the same for $\rho(\A)$ and $\rho(\A')$.

So, consider the tuples $\bar v_1,\ldots,\bar v_t$ that interpret the constant symbols
$c_1,\ldots,c_t$ in $\rho(\A)$.
Applying the claim again to the formulas $c_j = \bar v_j$ in the image $\rho(\A)$
and the tuple $\bar v=\tup{\bar v_1,\ldots,\bar v_t}$, we obtain a numeric
formula $\tilde{\beta}(\bar v)$ and non-numeric $\sigma$-literals
$\tilde{\lambda}_1(\bar v_1),\ldots,\tilde{\lambda}_{\ell'}(\bar v_t)$
with $\ell'\leq t$ such that 
\begin{displaymath}
\A\models \tilde{\beta}(\bar v)\land \tilde{\lambda}_1(\bar v)\land\cdots\land\tilde{\lambda}_{\ell'}(\bar v) \,.
\end{displaymath}
Further, if $w_1,\ldots,w_{t'}$ are the interpretation of the constants
$c'_1,\ldots,c'_{t'}$ in $\A$, then
\begin{displaymath}
\A\models \lambda_1(\bar u)\land\cdots\land\lambda_\ell(\bar u) \land
          \tilde{\lambda}_1(\bar v)\land\cdots\land\tilde{\lambda}_{\ell'}(\bar v) \land
          c'_1 = w_1 \land \cdots \land c'_{t'} = w_{t'}
\end{displaymath}
with $\ell + \ell' + t' \leq k$.
Thus, apply the $(n,k)$-uniformity of $S$ to obtain a structure $\A'\in S$ with 
$\|\A'\|=\|\A\|$ and such that
\begin{displaymath}
\A'\models \lambda_1(\bar u)\land\cdots\land\lambda_\ell(\bar u) \land
           \tilde{\lambda}_1(\bar v)\land\cdots\land\tilde{\lambda}_{\ell'}(\bar v) \land
           c'_1 = w_1 \land \cdots \land c'_{t'} = w_{t'} \,.
\end{displaymath}

In particular, $\A$ and $\A'$ have the same interpretation for the constants
$\{c'_1,\ldots,c'_{t'}\}$ and thus, by Proposition~\ref{prop:constantslemma},
$\A'\models\beta(\bar u)\land\tilde{\beta}(\bar v)$.
On the other hand, a new application of the claim but using the structure $\A'$,
gives us $\rho(\A')\models L_1(\bar u)\land\cdots\land L_m(\bar u)$ and
$\rho(\A')\models c_1=\bar v_1\land\cdots\land c_t=\bar v_t$. The latter
implies that $\rho(\A)$ and $\rho(\A')$ have the same interpretation for
the constants $\{c_1,\ldots,c_t\}$.
By Proposition~\ref{prop:constantslemma}, $\rho(\A')\models\theta(\bar u)$
as well. Therefore, $\rho(\A')\models\psi(\bar u)$ as needed.

\begin{proof}[Proof of the Claim:]
Let us consider the two cases whether $\eta(\bar y)$ is a positive or negative literal.
In the first case, $\eta(\bar y)$ is either $R(\bar y)$ or $c=\bar y$ for some
$R\in\{R_1,\ldots,R_s\}$ or $c\in\{c_1,\ldots,c_t\}$.
Suppose that $\rho(\B)\models\eta(\bar a)$. Then, there is a projective formula of the form
\begin{equation}\label{eq:projective}
\varphi(\bar y) =
  \alpha_0(\bar y) \lor (\alpha_1(\bar y) \land \lambda_1(\bar y)) \lor \cdots \lor 
                        (\alpha_p(\bar y)\land\lambda_p(\bar y))
\end{equation}
that defines the interpretation of $\eta(\bar y)$ in $\rho(\B)$ such that
$\B\models\varphi(\bar a)$. Then, either $\B\models\alpha_0(\bar a)$ which
is numeric or $\B\models\alpha_j(\bar a)\land\lambda_j(\bar a)$ for exactly
one $1\leq j\leq p$ since the formulas $\alpha_j$'s are mutually exclusive.
Conversely, if $\B\models\alpha_0(\bar a)$ or
$\B\models\alpha_j(\bar a)\land\lambda_j(\bar a)$ for some $1\leq j\leq p$,
then $\rho(\B)\models\eta(\bar a)$.

The second case is when $\eta(\bar y)$ is a negative $\tau$-literal
of the form $\eta(\bar y)=\neg R(\bar y)$. Let $\varphi(\bar y)$ be
the projective formula like \eqref{eq:projective} that defines the
interpretation of $R$ in the image of $\rho$.
Assume that $\rho(\B)\models\eta(\bar a)$. Then, $\B\nmodels\varphi(\bar a)$.
There are two possibilities.
First, $\B\nmodels\alpha_0(\bar a)\lor\cdots\lor\alpha_p(\bar a)$
in which case $\B\models\gamma(\bar a)$ for the numeric implicant
$\gamma(\bar y)=\neg\alpha_0(\bar y)\land\cdots\land\neg\alpha_p(\bar y)$.
Second, $\B\models\alpha_j(\bar a)\land\neg\lambda_j(\bar a)$ for
exactly one $1\leq j\leq p$.
In either case, the claim is satisfied.
Finally, for the converse direction, if $\B\models\gamma(\bar a)$
or $\B\models\alpha_j(\bar a)\land\neg\lambda_j(\bar a)$, then 
$\rho(\B)\nmodels R(\bar a)$ and hence $\rho(\B)\models\eta(\bar a)$.
\end{proof}
\end{proof*}

\begin{lem}\label{le:corelemma2}
Let $\sigma$ and $\tau=\tup{R_1,\ldots,R_s,c_1,\ldots,c_t}$ be two vocabularies
where $\sigma$ has constants $\{c'_1,\ldots,c'_{t'}\}$, and let $n\geq 0$
be a non-negative integers.
Further, let 
\begin{description}
  \item[H1] $S\subseteq\struc(\sigma)$ be a subset of structures that contains all the
    structures $\A$ with $\|\A\|<n$,
  \item[H2] $\psi=\forall\bar x\theta(\bar x)$ be a \FO-sentence with
    $\theta(\bar x)\in\cnf_r(\tau)$ for some integer $r$, and
  \item[H3] $\Phi$ be a sentence in $\L(\tau)$. 
\end{description}
If $S$ is $(n,k)$-uniform for $k\geq t+t'+r$ and $\rho$ is a fop that \emph{reduces}
$S$ to $\mod(\Phi\land\psi)$, then $\psi$ is \emph{superfluous} with respect to $\rho$.
\end{lem}
\begin{proof}
Assume that $\theta(\bar x)= \bigwedge_{1\leq i\leq m}\theta_i(\bar x)$
where each $\theta_i$ is a clause of the form
\[ \theta_i(\bar x) \equiv \beta_i(\bar x)\lor L_1^i(\bar x)\lor\cdots\lor L_{m_i}^i(\bar x) \,, \]
where $\beta_i$ is a disjunction of numeric literals
and $L_1^i(\bar x),\ldots,L_{m_i}^i(\bar x)$ are non-numeric literals with $m_i\leq r$.
We want to show that $\rho(\A)\models\forall\bar x\theta(\bar x)$ for
every $\A\in\struc(\sigma)$.
Let us consider two cases:
\begin{itemize}
\item If $\A\in S$, then $\rho(\A)\models\psi$ since $\rho$ reduces $S$ to $\mod(\Phi\land\psi)$.
\item If $\A\not\in S$, then $\|\A\|\geq n$ and $\rho(\A)\not\in\mod(\Phi\land\psi)$.
  Thus, either $\rho(\A)\nmodels\Phi$ or $\rho(\A)\nmodels\psi$.
  Assume, for the sake of a contradiction, that $\rho(\A)\not\models\psi$.
  So, there is a clause $\theta_i(\bar x)$ and a tuple $\bar u\in |\rho(\A)|^{|\bar x|}$
  such that $\rho(\A)\not\models\theta_i(\bar u)$; i.e.,
  \begin{alignat*}{1}
    \rho(\A) & \nmodels \beta_i(\bar u)\lor L_1^i(\bar u)\lor\cdots\lor L_{m_i}^i(\bar u)
  \intertext{or equivalently}
    \rho(\A) & \models \gamma(\bar u)\land \lambda_1(\bar u)\land\cdots\land\lambda_{m_i}(\bar u)
  \end{alignat*}
  where $\gamma(\bar x)\equiv\neg\beta_i(\bar x)$ is a conjunction of numeric literals,
  and each $\lambda_j(\bar x)\equiv\neg L_j^i(\bar x)$ is a non-numeric literal.
  At this moment, all the hypotheses of Lemma~\ref{le:corelemma1} are satisfied and
  thus there is a finite $\sigma$-structure $\A'\in S$ with $\|\A'\|=\|\A\|$ such that
  $\rho(\A')\models\gamma(\bar u) \land \lambda_1(\bar u) \land\cdots\land \lambda_{m_1}(\bar u)$.
  Hence, $\rho(\A')$ does not satisfy $\psi$.
  This is a contradiction since $\A'\in S$ and $\rho$ reduces $S$ to $\mod(\Phi\land\psi)$.
  Therefore, $\rho(\A)\models\psi$. \qedhere
\end{itemize}
\end{proof}

\begin{cor}\label{cor:corelemma2}
With the same hypotheses of Lemma~\ref{le:corelemma2}, the fop $\rho$ reduces $S$
to $\mod(\Phi)$. Hence, if $\L$ captures $\C$ and $S$ is $\C$-complete, then
$\mod(\Phi)$ is $\C$-complete.
\end{cor}
\begin{proof}
Direct since for every structure $\A\in\struc(\sigma)$, $\rho(\A)\models\psi$.
\end{proof}

We now define what is a complete and uniform family of problems for a class $\C$,
and state and prove the main theorem of the paper.

\begin{defi}
A family $\F$ of problems over vocabulary $\sigma=\tup{R_1,\ldots, R_s,c_1,\ldots,c_t}$ 
is \emph{complete and uniform} for a complexity class $\C$ if
1)~every problem in $\F$ is $\C$-complete, and 
2)~there is a sequence $\set{n_k}_{k\geq0}$ and a natural number $m$ such that
for every $k\geq m$ there is a $(n_k,k)$-uniform problem $S_{n_k}$ in $\F$ that
contains all the structures $\A\in\struc(\sigma)$ with $\|\A\|<n_k$.
\end{defi}

\begin{thm}[Main]\label{theo:main}
Let $\C$ be a complexity class captured by $\L$ with $\FO\subseteq\L$.
If $\C$ contains a complete and uniform family $\F$,
then \FOA is superfluous with respect to $\C$.
\end{thm}
\begin{proof}
Let $\tau$ be a vocabulary and $\psi$ be a \FOA-sentence on $\tau$. 
It is enough to prove that $\psi$ is superfluous with respect to $\L$
since $\psi$ is an arbitrary sentence.
First notice that $\psi$ can be written in prenex normal form with a
quantifier-free part in $\cnf_r$ for some $r\in\mathbb N$.
Let $\Phi\in\L(\tau)$ be a sentence such that $\mod(\Phi\land\psi)$ is
$\C$-complete.
Let $t$ be the number of constant symbols in the vocabulary $\sigma$ for
the family $\F$, $k=r+t+t'$ where $t'$ is the number of constant symbols
in $\tau$, and $\{n_k\}_{k\geq 0}$ and $m$ be the sequence and natural
number for $\F$.
If $k\geq m$, there is a problem $S_{n_k}\in\F$ that is $(n_k,k)$-uniform,
$\C$-complete, and contains all the structures $\A\in\struc(\sigma)$
with $\|\A\|<n_k$. Then, there is a fop $\rho$ that reduces $S_{n_k}$
to $T$.
Thus, all the hypotheses in Lemma~\ref{le:corelemma2} are fulfilled and
therefore $\mod(\Phi)$ is $\C$-complete by Corollary~\ref{cor:corelemma2}.
The case $k<m$ is covered by the case $k=m$ since $(n_k,k)$-uniformity 
implies $(n_k,k-1)$-uniformity according to Lemma \ref{le:prop1-uniformity}.
\end{proof}

\subsection{Superfluity of \FOA for Some Complexity Classes} 

We have seen that {\sc 0m-HP} is $(4k,k)$-uniform for every $k\in\mathbb N$,
but it easy to see that it has negative instances of every size and thus
the Theorem~\ref{theo:main} cannot be applied directly.
Hence, we make the following definitions:

\begin{defi}
If $S$ is a problem over $\sigma$, we define for each $n\in\mathbb N$:
\[ S_n := S \cup \setdef{\A\in\struc(\sigma)}{\|\A\|<n} \]
and the family of problems
\[ \F(S) := \set{S_n}_{n\geq 2} \,. \]
\end{defi}

We will also use another property.
The notion of autoreducibility is well known \cite{buhrman:structure}.
We can translate it in the context of fops by saying that
a problem is autoreducible if there is a reduction from it to itself 
different than the identity.
We need autoreducible sets with an extra requirement on the
cardinalities of the image structures:

\begin{defi}
Given a vocabulary $\sigma$ and a natural number $n$, a set
$S\subseteq\struc(\sigma)$ is \emph{$n$-autoreducible} if
there is a fop $\rho:\struc(\sigma)\rightarrow\struc(\sigma)$
which reduces $S$ to itself and such that $\|\rho(\A)\|>n$
for every $\A\in\struc(\sigma)$.
\end{defi}

It is immediate to see that the problem $S_n$ is $\C$-hard if
$S$ is $\C$-hard and $n$-autoreducible.

\begin{thm}\label{theo:supinnp}
$\F(\textsc{Reach})$, $\F(\textsc{AltReach})$,
  $\F\text(\textsc{0m-HP})$ and $\F(\textsc{CoMonoTriangle})$ are
  complete and uniform families for \NLSPACE, \PTIME, \NP and \coNP
  respectively.
\end{thm}

Then, as a consequence of Theorems \ref{theo:main} and \ref{theo:supinnp}:
\begin{cor}\label{co:superfluity}
\FOA is superfluous with respect to \NLSPACE, \PTIME, \NP, and \coNP
\end{cor}
This corollary answers Conjecture \ref{conj:medina-full} for \FOA instead of $\FO$.

\begin{proof}[Proof of Theorem~\ref{theo:supinnp}]
We begin showing that the families $\F(\textsc{Reach})$,
$\F(\textsc{AltReach})$, $\F(\textsc{0m-HP})$ and
$\F(\textsc{CoMonoTriangle})$ are uniform.

By Lemma~\ref{le:uniformity-examples}, the problems in
$\F(\textsc{Reach})$ and $\F(\textsc{AltReach})$ are $(2k+1,k)$
uniform, the problems in $\F(\text{0m-HP})$ are $(4k,k)$-uniform, and
the problems in $\F(\textsc{CoMonoTriangle})$ are $(2k+6,k)$-uniform
for every integer $k\geq0$.  Thus, it is easy to these that these
families are uniform: for $\F(\textsc{Reach})$ and
$\F(\textsc{AltReach})$ the sequence is $\set{2k+1}_{k\geq0}$, for
$\F(\textsc{0m-HP})$ the sequence is $\set{4k}_{k\geq0}$, and for
$\F(\textsc{coMonoTriangle})$ the sequence is
$\set{2k+6}_{k\geq0}$. In all cases, $m=1$.

It remains to show that the families are complete; i.e., that every
problem in the families $\F(\textsc{Reach})$, $\F(\textsc{AltReach})$,
$\F\text(\textsc{0m-HP})$ and $\F(\textsc{CoMonoTriangle})$ is complete
for the classes \NLSPACE, \PTIME, \NP and \coNP respectively. 
That is, that every problem belongs and is hard for the respective
complexity class.

Let us first show that each problem in the above families belongs
to the respective complexity class.
Consider the family $\F(S)$ where $S$ is a problem in a complexity
class $\C$ captured by the logic $\L$ with $\FO\subseteq\L$.
There is a sentence $\Phi\in\L$ such that $S=\mod(\Phi)$.
On the other hand, given integer $n\geq0$, we have a \FO-sentence
$\zeta_n$ such that $\A\models\zeta_n \iff \|\A\|<n$.
Thus, $\Phi\lor\zeta_n$ defines $S_n$ and, since $\FO\subseteq\L$,
$\Phi\lor\zeta_n\in\L$ and $S_n$ belongs to $\C$.
Therefore, every problem in $\F(S)$ belongs to $\C$.
Since the logics that capture the classes \NLSPACE, \PTIME, \NP
and \coNP all include FO, then each problem in each family
belongs to the respective complexity class.

The hardness for each problem in the families $\F(S)$, when $S$ is
{\sc Reach}, {\sc AltReach}, {\sc 0m-HP} or {\sc coMonoTriangle},
follows from the facts that $S$ is hard for its complexity class
and that $S$ is $n$-autoreducible for each integer $n\geq0$.
We begin by showing the latter fact.  That is, given integer $n\geq0$,
we need to construct a fop-reduction $\rho$ such that the image
$\rho(\A)$ has a universe of size greater than $n$.
The reduction in all cases essentially consists 
of padding the original structure to obtain a new one
with the desired size.
Given an integer $n\geq 0$ we let $k$ be the least natural number
such that $2^k>n$. The integer $k$ will be the \emph{arity} of the
fop $\rho$ in each case.

{\sc AltReach} is defined
over the vocabulary $\sigma=\tup{E,U,s,t}$ where $E$ and $U$ are
a binary and monadic relations denoting the edges and universal
vertices in the graph, and $s$ and $t$ are constant symbols
denoting designated vertices.
In this case the reduction simply adds as much disconnected vertices
as necessary.

If $\A=\tup{|\A|,E^\A,U^\A,s^\A,t^\A}$ is a finite $\sigma$-structure.
The image of $\A$ is the finite structure $\rho(\A)$ defined as follows:
\begin{alignat*}{1}
|\rho(\A)| &= |\A|^k\\
E^{\rho(\A)}
  &= \setdef{(\bar u,\bar v)}{\bigwedge_{1\leq j<k}(u_j=v_j=0) \land (u_k,v_k) \in E^\A} \\
U^{\rho(\A)}
  &= \setdef{\bar v}{\bar v\!\upharpoonright_{k-1}=\bar 0\land v_k\in U^\A} \\
s^{\rho(\A)}
  &= \tup{0,\ldots, 0,v} \quad\text{with}\quad v=s^\A \\
t^{\rho(\A)}
  &= \tup{0,\ldots, 0,v} \quad\text{with}\quad v=t^\A
\end{alignat*}
it should be clear that $\rho$ is a reduction
since $\A$ contains an alternating path from $s$ to $t$ iff $\rho(\A)$
contains one.
The projection reducing {\sc Reach} to {\sc Reach$_n$} is almost the same but without
any reference to universal vertices.

In the case of $\F(\textsc{0m-HP})$ 
given a strucure $\A$ its image $\rho(\A)$ 
consists of $k$ copies of $\A$ with edges joining the vertex
corresponding to $\max$ in the $j$-th 
with the vertex corresponding to $0$ in the $j+1$-th copy with 
$j<n$. 
There are no other edges connecting different copies of $\A$.
It is clear that $\rho(\A)$ will have a Hamiltonian path joining
the vertices $\tup{0,\ldots, 0}$ and $\tup{\max,\ldots, \max}$
iff there is a Hamiltonian path joining $0$ and $\max$ in $\A$ 

For {\sc coMonoTriangle} it is enough to have the $\|\A\|^{k-1}$
copies of the input structure without connecting them. Since they
are all different connected components of $\rho(\A)$, a 2-coloring
of the edges in $\rho(\A)$ is just a combination of 2-colorings of
the edges in $\A$ and $\rho(\A)$ is a positive instance of
{\sc CoMonoTriangle} if and only if $\A$ is.

We finish the proof of the theorem by showing that each of the problems
is complete for its complexity class.  We already know that {\sc Reach}
is \NLSPACE-complete, {\sc AltReach} is \PTIME-complete \cite[Corollary 11.3]{immerman:book}
and {\sc coMonoTriangle} is \coNP-complete since {\sc MonoTriangle} is
\NP-complete \cite{medina:thesis,Garey&Johnson}.
It remains to show that {\sc 0m-HP} is \NP-complete.
For the inclusion, notice that {\sc 0m-HP} can be defined in \SOE 
by stating a total ordering of the vertices such that any two 
consecutive vertices in the ordering form an edge in the graph 
and with the first and last element in the ordering being $0$
and $\max$ respectively.
For the hardness, it is enough to reduce the similar {\sc 01-HP}
problem (known to be \NP-complete \cite{borges-bonet12}) to
{\sc 0m-HP} with a projection that interchanges $1$ with $\max$.
\end{proof}

\section{Applications}\label{sec:applications}

We show two applications of the superfluity of \FOA.
In the first application, we establish the \NP-completeness of the
problem {\sc LongestPath} for determining the existence of a path
between two designated vertices $s$ and $t$ of length bigger than a
given threshold $K$.
In the second application, we show that establishing the completeness
of a property for undirected graphs is enough for establishing
the completeness of the same property for directed graphs.
Up to our knowledge, these are novel results in the area.

\begin{thm}\label{theo:applications1}
{\sc LongestPath} is \NP-complete via fops.
\end{thm}

Before giving the proof of the theorem, let us define some terminology.
Let $S$ be a problem where each instance includes a function $f$
defined from the set of $k$-tuples of elements in the universe
to the non-negative integers.
If $\sigma$ is the vocabulary used to define $S$, we say that a $k+1$-ary
relation symbol $F$ in $\sigma$ is the \emph{binary expansion of $f$}
if, given a finite $\sigma$-structure $\A$ and a $k$-tuple $\bar a\in\|\A\|^k$:
\[ (a_1,\ldots, a_k,i)\in F^\A \iff \text{the $i$-th bit in the binary expansion of $f(\bar a)$ is 1.} \]

\begin{proof}
{\sc LongestPath} is known to be in NP \cite{Garey&Johnson}, hence there
is a \SOE sentence $\Phi_{\text{\sc LP}}$ for it over the vocabulary
$\sigma=\tup{L^3,E^2,K^1,s,t}$, where $L$ defines the binary expansion of
the lengths of the edges, $E$ is the edge relation and $K$ defines the
binary expansion of the bound on the total length for a path joining
$s$ and $t$.

Consider the \FOA sentence $\psi=\psi_1\land\psi_2$ where 
\begin{alignat*}{1}
\psi_1 	&= \forall xyz\, L(x,y,z) \longrightarrow z = 0 \\
\psi_2	&= \forall x \, K(x) \longleftrightarrow \bit(\max,x) \,.
\end{alignat*}
Then a finite $\sigma$-structures $\A$ is a model of the sentence
$\Phi=\Phi_{\textsc{LP}}\land\psi$ if and only if every edge of
$\A$ has length 1 and there is a path from $s^\A$ to $t^\A$ with
total length at least $\max$. 

Since $\max=\|\A\|-1$ and every edge has length 1, the path joining
$s$ and $t$ must be Hamiltonian.
Therefore $\Phi=\Phi_{\textsc{LP}}\land\psi$ defines 
{\sc HamiltonianPathBetweenTwoPoints} (which is \NP-Complete \cite{borges-bonet12})
and $\Phi_{\textsc{LP}}$ defines an \NP-complete problem because
$\psi$ is superfluous for \NP.
\end{proof}

The second application considers directed and undirected versions of problems
on graphs. If $S\subseteq\struc(\tup{E^2})$ is a problem over undirected graphs,
we say that $\tilde S$ is the \emph{version over directed graphs} of $S$ if 
there is a sentence $\Phi$ such that $\tilde S=\mod(\Phi)$ and $S=\mod(\Phi\land\psi)$
where $\psi\equiv\forall xy\, E(x,y)\longrightarrow E(y,x)$. In such case,
if $S$ is $\C$-complete for a class $\C$ for which \FOA is superfluous, then
$\tilde S$ is also $\C$-complete.

\begin{thm}\label{th:dirvsundir}
If \FOA is superfluous with respect to class $\C$ and $S$ is a $\C$-complete
problem over undirected graphs, then its version over directed graphs is
$\C$-complete as well.
\end{thm}

\section{Discussion}
\label{sec:conclusions}

We gave a partial affirmative answer to Conjecture~\ref{conj:medina-full} by
proving that the universal fragment \FOA is superfluous with respect to \NP.
This is a consequence of Theorem~\ref{theo:main} which also implies that
\FOA is superfluous with respect to \NLSPACE, \PTIME and \coNP. Our method
is fairly general and it is based on the $(n,r)$-uniformity concept which
is new, as far as we know. On the other hand, we extended the previous
concept of superfluity \cite{medina:thesis} to complexity classes
(and languages) beyond \NP.

In one application of these results, the superfluity of \FOA allows to give
syntactic proofs for \NP-completeness if one can find a suitable restriction
of the problem that is expressible with a \FOA sentence and already known
to be complete (cf.\ Theorem~\ref{theo:applications1}).
In another application we show that the directed versions of any \NP-complete
problem on undirected remain \NP-complete.

In the future, we want to continue this research with the aim of 
proving Medina's conjecture in its full generality, which we 
believe to be true.
On the other hand, one can think in generalizations and related
versions of this conjecture. For example, a more general conjecture
is the following:

\begin{conj}\label{conj:generalizedconjecture}
Let $\L$ and $\L'$ be two logics capturing classes $\C$ and $\C'$ respectively
and such that $\L'\subseteq \L$, and let $\Phi\in\L$ and $\Phi'\in\L'$ be two
sentences. Then,
\[
  \text{$\mod(\Phi\land\Phi')$ is $\C$-complete} \ \implies \
  \text{either $\mod(\Phi)$ or $\mod(\Phi')$ is $\C$-complete} \,.
\]
\end{conj}
This conjecture reduces to Medina's when $\L'$ is \FOA
and \AC (languages recognized by circuits of polynomial size, constant
depth and unbounded fan-in) is known to be strictly included in $\C$.

Another direction is to use the concept of complete and uniform families
to separate complexity classes. That is, if $\C$ and $\C'$ are two
complexity classes such that $\C\subseteq\C'$ and $\C'$ contains
a complete and uniform family $\F$ but $\C$ does not contain such
a family, then it must be the case that the two classes are different.

The converse of Medina's conjecture is also interesting. In general, we
know that the completeness of $\mod(\Phi)$ does not necessarily imply the completeness
of $\mod(\Phi\land\psi)$. Indeed, it is enough for $\psi$ to be inconsistent
to see this or, for example, consider the case of \SAT and \TwoSAT in which
the first is \NP-complete while the second is in \PTIME and not believed to
be \NP-complete, yet \TwoSAT can be expressed as $\mod(\Phi_\SAT\land\psi)$
for a suitable choice of $\psi$ where $\Phi_\SAT$ defines \SAT.
However, an interesting question is what syntactic characteristics must have
$\psi\in\FO$ in order for the $\C$-completeness of $\Phi$ to be preserved by
the conjunction $\Phi\land\psi$.

Finally, there is a clear relation between the concept of $(n,r)$-uniformity
and Ramsey-type problems.
Let $S$ be a $(n,r)$-uniform class of graphs defined over the vocabulary  
$\tup{E^2}$ and $m\geq n$.
Then, a sequence $\set{L_j(x_j,y_j)}_{1\leq j\leq r}$ of literals together with
a sequence $\set{(a_j,b_j)}_{1\leq j\leq r}$ of different pairs from $[m]^2$, with
$a_j\neq b_j$ for $1\leq j\leq r$, can be interpreted as a \emph{partial 2-coloring}
on the edges of $K_m$, the complete graph on $m$ vertices, by considering the edge
$(u,v)$ to be colored red or blue whether $E(u,v)$ or $\neg E(u,v)$ belong to the
sequence of literals respectively.
Since $S$ is $(n,r)$-uniform, there there is a coloring $c$ of the edges of $K_m$
such that the subgraph consisting of all the $m$ vertices but only the red edges
belongs to $S$.
This connection may prove useful when giving a definite answer to Medina's
conjecture.

\section*{Acknowledgments}
\noindent The first author would like to thank Prof.\ Argimiro Arratia, his
former thesis advisor. Though Prof.\ Arratia did not collaborate
directly in this paper, this work was motivated to a great extent
by fruitful mathematical discussions with him.
We are in debt to the reviewers and editor whose comments helped
us to improve the paper.

\bibliographystyle{alpha}
\bibliography{paper}

\appendix

\section{Decision Problems}\label{app:problemsdefinitions}


\small
\smallskip\bigskip

\centering
\fbox{
\begin{minipage}{.9\textwidth}
{\sc 3-DimensionalMatching}\\[-1.7em]
\begin{center}
\begin{description}
\item[\quad Instance:] a collection $M$ of triplets over a set $S$.
\item[\quad Property:] the existence of a three dimensional matching $M'$ contained in $M$;
  i.e., a subset $M'\subseteq M$ with $|M'|=|S|$ such that for every $a\in S$,
  there are exactly three triplets $(a,y,z)$, $(x',a,z')$ and $(x'',y',a)$ in $M'$.
\item[\quad Vocabulary:] $\sigma=\tup{M^3}$.
\end{description}
\end{center}
\end{minipage}}

\smallskip\bigskip

\centering
\fbox{
\begin{minipage}{.9\textwidth}
{\sc AltReach}\\[-1.7em]
\begin{center}
\begin{description}
\item[\quad Instance:] an alternating graph $G$ with two highlighted vertices $s$ and $t$.
\item[\quad Property:] the vertex $t$ is accessible from the vertex $s$.
\item[\quad Vocabulary:] $\sigma=\tup{E^2,U^1,s,t}$.
\end{description}
\end{center}
\end{minipage}}

\smallskip\bigskip

\fbox{
\begin{minipage}{.9\textwidth}
{\sc HamiltonianPathBetweenTwoPoints}\\[-1.7em]
\begin{center}
\begin{description}
\item[\quad Instance:] a finite simple graph $G$ with two special vertices $s$ and $t$.
\item[\quad Property:] existence of a Hamiltonian path between $s$ and $t$.
\item[\quad Vocabulary:] $\sigma=\tup{E^2,s,t}$.
\end{description}
\end{center}
\end{minipage}}

\smallskip\bigskip

\fbox{
\begin{minipage}{.9\textwidth}
{\sc HamiltonianPathBetweenZeroAndMax}\\[-1.7em]
\begin{center}
\begin{description}
\item[\quad Instance:] a finite simple graph $G$ with $\set{0,\ldots, n-1}$ as its set of vertices.
\item[\quad Property:] existence of a Hamiltonian path between $0$ and $n-1$.
\item[\quad Vocabulary:] $\sigma=\tup{E^2}$.
\end{description}
\end{center}
\end{minipage}}

\smallskip\bigskip

\fbox{
\begin{minipage}{.9\textwidth}
{\sc LongestPath}\\[-1.7em]
\begin{center}
\begin{description}
\item[\quad Instance:] a finite simple graph $G$ with lengths $\ell(e)\in\mathbb Z^+$
  associated to each edge, two special vertices $s$ and $t$, and an lower bound $K\in\mathbb Z^+$.
\item[\quad Property:] existence of a simple path between $s$ and $t$ with length at least $K$.
\item[\quad Vocabulary:] $\sigma=\tup{L^3,E^2,K^1,s,t}$.
\end{description}
\end{center}
\end{minipage}}

\smallskip\bigskip

\fbox{
\begin{minipage}{.9\textwidth}
{\sc MonochromaticTriangle}\\[-1.7em]
\begin{center}
\begin{description}
\item[\quad Instance:] A graph $G$.
\item[\quad Property:] There is a 2-coloring of the edges of $G$ such that $G$
  contains no triangle with all edges of the same color.
\item[\quad Vocabulary:] $\sigma=\tup{E^2}$.
\end{description}
\end{center}
\end{minipage}}

\smallskip\bigskip

\fbox{
\begin{minipage}{.9\textwidth}
{\sc Reach}\\[-1.7em]
\begin{center}
\begin{description}
\item[\quad Instance:] A graph $G$ with two highlighted vertices $s$ and $t$.
\item[\quad Property:] There is a path between $s$ and $t$.
\item[\quad Vocabulary:] $\sigma=\tup{E^2,s,t}$.
\end{description}
\end{center}
\end{minipage}}

	
	
\vspace{-30 pt}
\end{document}